\documentclass[reqno,keywordsasfootnote,12pt]{article}

\usepackage[english]{babel}
\usepackage[latin1]{inputenc}
\usepackage{amsmath,amsthm,amsfonts,amssymb,times}
\numberwithin{equation}{section}
\usepackage[final]{epsfig}
\usepackage{color}
\usepackage{cancel}

\usepackage{dsfont}

\def\1{ \mathds{1} }

\usepackage{mathrsfs}

\usepackage{enumerate}

\makeatletter
\def\author#1{\gdef\autrun{\def\and{\unskip, }#1}\gdef\@author{#1}}

\makeatother

\newtheorem{theorem}{Theorem}[section]

\newtheorem{lemma}[theorem]{Lemma}
\newtheorem{proposition}[theorem]{Proposition}

\theoremstyle{definition}
\newtheorem{definition}[theorem]{Definition}

\numberwithin{equation}{section}

\usepackage{graphicx,pgf}

\makeatletter

\DeclareMathOperator{\sgn}{sgn}

\newcommand{\C}{\mathbb{C}}
 
\newcommand{\G}{\mathcal{G}}

\def\be{\begin{equation}}
\def\ee{\end{equation}} 
\newcommand*{\Chi}{\mbox{\Large$\chi$}}% big chi 
\DeclareMathOperator{\Pf}{Pf}  %%% Pfaffian
\newcommand{\m}{{\bf m}}
 
\makeatletter
\newcommand\xleftrightarrow[2][]{%
  \ext@arrow 9999{\longleftrightarrowfill@}{#1}{#2}}
\newcommand\longleftrightarrowfill@{%
  \arrowfill@\leftarrow\relbar\rightarrow}
\makeatother

\newcommand*{\connect}[2]{#1\leftrightarrow #2}
\newcommand*{\Gconnect}[2]{#1{\xleftrightarrow{\omega^{(2)}}} #2}
\usepackage{hyperref}
\usepackage{tikz}

\makeatother

\begin{document}

%%%%% To ease editing, add:

%\baselineskip=17pt

%%%%%%%%%%%%%%%%

%% In the running head, give an abbreviation of the title. 
%\titlerunning{Resonant delocalization}

\title{\LARGE Pfaffian Correlation Functions  \\
of Planar Dimer Covers } 
%Through\\ a Graph Zeta Function }
 
%\title{\LARGE 2D Ising model through\\
% Graph Zeta Functions } 

\author{Michael Aizenman \footnote{Departments of Mathematics and Physics, Princeton University, USA.}
\and  Manuel La\'inz Valc\'azar \footnote{Princeton visiting student from Univ.\ Cantabria, Spain. }   
\and Simone Warzel \footnote{Zentrum Mathematik, TU M\"unchen, Germany.} }

\date{\small Version of \today \\[4ex]
%\hfill \begin{minipage}{6cm} \it To Yakov G. Sinai and David Ruelle 
%\end{minipage}
}

\maketitle

%\subjclass{Primary ...; Secondary ....}

%%%%%%%%

\begin{abstract}

The  Pfaffian structure of the boundary monomer correlation functions in the dimer-covering planar graph models is rederived through a combinatorial / topological argument.   
These  functions are then extended into a larger family of order-disorder correlation functions which are shown to exhibit Pfaffian structure throughout the bulk.  Key tools involve combinatorial switching symmetries which are identified through the loop-gas representation of  the double dimer model, and topological implications of planarity.    

%%% Keywords are optional
%\keywords{ }

\end{abstract}

%%%%%%%%%  

%\setcounter{tocdepth}{2} 
%{\small
%\tableofcontents 
%}
%\vfill

%\newpage

\section{Introduction}

%The combinatorial problem of enumeration of dimer covers (aka domino tilings)  of graphs has been attracting much attention 
%from diverse scientific communities ranging from pure mathematics to its origin in statistical physics. 
%Starting with the groundbreaking papers by  P.~W.~Kasteleyn,  M.~E.~Fisher and H. V. N.~Temperley~\cite{Kas61,Fisher_Dimer61,TF61,Kas63}, 
%it is known that for planar  graphs this problem admits a simple solution in terms of a Pfaffian of what is now known as the Kasteleyn matrix.  
%However, the general planar dimer-cover problem with fixed density of vacancies (aka monomers) is computationally hard and thus not of simple Pfaffian form \cite{Jer87}. 
% 
 
The combinatorial problem of enumeration of dimer covers  of graphs (aka domino tilings)  has attracted interest from a diverse range of perspectives.  These include  statistical mechanics, combinatorics, and algorithm complexity studies. 
%statistical mechanics, algebraic and Grassmannian integration methods for combinatorics, and algorithm complexity studies.   
% attracting much attention 
%from diverse scientific communities due to its interest in statistical physics to combinatorics, and algorithm complexity theory. 
In their groundbreaking papers, P.~W.~Kasteleyn,  M.~E.~Fisher and H. V. N.~Temperley~\cite{Kas61,Fisher_Dimer61,TF61,Kas63}, 
showed that for planar  graphs the pure dimer problem admits a simple solution in terms of a Pfaffian of what is now known as the Kasteleyn matrix.  
The pure dimer partition functions is different in this sense from its monomer-dimer extension, for which its evaluation is computationally hard and thus not of simple Pfaffian form \cite{Jer87}.

Extensive research has been devoted  to various facets of dimer coverings,
specially in the case of planar and bipartite graphs. Examples include the close relation between the partition functions of the dimer cover and of the Ising model~\cite{Kas63,Fisher_Ising,PY80}, non-existence 
of phase transitions~\cite{HeilLieb}, structure of the model's correlation functions,   the arctic circle phenomenon~\cite{CEP96},  
continuum limits and their description in terms of (conformal) field theory.  More on this may be found in the overviews~\cite{KenyonPC,Dik,All15} and references therein.\\

%It was also noted that for planar graphs the boundary $2n$-monomer correlation functions (whose explicit definition is restated below)  are given by Pfaffians of the corresponding $2$-point functions~\cite{PRue08,GJL16}.  The first point presented here is a simple  geometric explanation of this fact, which is based on a combinatorial relation, akin to the Ising model's switching lemma, combined with a simple topological argument.   
%%
%%%The main aim of this note is to give a simple geometric explanation of the Pfaffian nature of monomer-dimer correlation functions.
%%%Other than planar and of bounded degree the graphs covered can be arbitrary. 
%%%We thereby simplify 
%%
%The  argument presented here also makes it clear that the Pfaffian structure does not extend to the monomer correlation functions in the bulk.  \\ 

Our main aim here is to present a simple geometric explanation of the Pfaffian nature of some of the model's correlation function, through which the model's $2n$ correlation functions can be determined from just the corresponding two point function.   
The  proofs given here bear close similarity to  the methods which have recently been developed  for planar Ising spin models~\cite{ADTW}.  In analogy to the latter, the method relies on a combinatorial relation, which is valid for general graphs, combined with topological properties of planar graphs.       

It was already noted that for planar graphs the boundary  monomer correlation functions, whose explicit definition is restated below,  are given by Pfaffians of the corresponding $2$-point functions~\cite{PRue08,GJL16}.   The relation is less simple for the bulk monomer correlation functions, but it was pointed out that these can be written as products of two Pfaffians~\cite{AllF14}.    \\

We start by giving an elementary geometric proof of the Pfaffian structure of the boundary monomer functions.  The derivation also  explains why these functions do not have the Pfaffian structure in the bulk.  Furthermore, we formulate more explicitly than was done in the literature the model's disorder operators, and show that the expectation values of products of order-disorder operators yield correlation functions which are simultaneously Pfaffian throughout the bulk and reduce to the simper monomer correlation functions for sites along a boundary line.  

% hold for 
% is given an elementary proof, which also expel we show that a suitable combination of monomer insertions and certain  ``disorder operators'' yields correlation functions which have the Pfaffian structure in the bulk, and which for boundary sites reduce to the previously discussed monomer functions.\\ 

The disorder operators can be viewed as incomplete implementations of the dimer model's $Z_2$ gauge symmetry.  
From this perspective, their construction and basic properties are similar to those of the corresponding concept for the Ising model, as discussed by L.~P.~Kadanoff and  H.~Ceva~\cite{Kad_Ceva}.   \\

The  combinatorial and topological arguments presented here parallel the analogous discussion of planar Ising model  in the introductory sections of  \cite{ADTW}. 
An essential  tool is a  path integral representation of a duplicated system, which is referred to as   
 the  double dimer model.  The latter has been studied by R.~Kenyon and D.~Wilson (cf.\ \cite{KenyonWilson,Kenyon14} and references therein) and is related to the monopole-dimer model recently studied in~\cite{Ay15}.  \\ 

%The path
% integral in this paper concerns a sum over non-intersecting paths and loops and has the benefit of rendering available topological arguments, which are at the core of our proofs. 
%Different functional integral representations based on Grassmann integrals have been employed previously for basic results on this model, including  the presentation of the 
%monomer correlation functions in the bulk  in terms of  products of Pfaffians~\cite{AllF14}.\\

 % 
%
\section{Dimer covers and monomer correlations}  

Given a finite graph $\G = ( \mathcal{V} , \mathcal{E}) $ of vertex set $ \mathcal{V} $, a perfect matching or dimer cover is a 
subset of the edge set, $\omega \subset \mathcal E$,  
such that every vertex is covered by exactly one edge.  
The set of perfect
matchings is denoted $\Omega_ \G $.  
The dimer-cover partition function counts the number of
the graph's perfect matchings.   

Perfect matchings can also be weighted through a complex-valued edge function 
$K \,: \mathcal E \mapsto \C$.      
Given such an edge weight, 
the weighted dimer-cover partition function is 
\be 
Z_{ \G, K}  \ := \  \sum_{\omega\in \Omega_\G} \chi_K(\omega) 
\ee 
with 
$$
\chi_K(\omega)  := \prod_{b\in \omega} K_b \, . 
$$

Of particular interest is the effect on the dimer-cover partition function of 
the removal of a collection of sites, $ M \subset  \mathcal{V}$, which are regarded 
as covered by separate monomers.   The collection of perfect matchings of the remaining vertices is
denoted by $\Omega_\G(M) $ and 
\be 
Z_{\G,K}(M) \ := \   \sum_{\omega\in \Omega_\G(M)} \chi(\omega) 
\ee
stands for the weighted partition function of the monomer-depleted graph. 
It should be noted that  not all graphs admit  a perfect matching.  
In particular  if $M$ is of odd cardinality, then at least one of the factors in $ Z_{\G,K}\times Z_{\G,K}(M) $ vanishes. For simplicity, we shall concentrate in this paper on the case $ Z_{\G,K} \neq 0 $ 
for which the \emph{monomer correlation function} for an even collection of disjoint sites 
$\{ x_1,..., x_{2n}\} \subset \mathcal{V} $  is well-defined as 
\be \label{M_ratio} 
S_{2n}(x_1,..., x_{2n}) := \langle \prod_{j=1}^n \eta_{x_j} \rangle_{\G,K} \ := \  \frac{Z_{\G,K}(\{ x_1,..., x_{2n}\} )} {Z_{\G, K}}
\ee 
The variables $ \eta_{x_j} $ should be thought of an operator in the functional integral representing the average $ \langle \cdot \rangle_{\G,K} $ corresponding to the dimer partition function~$ Z_{ \G, K}  $.
These variables take a similar role as the spin variables in the related Ising model. 

In the planar set-up, monomer correlations have been studied early on by M.~E. Fisher and J.~Stephenson~\cite{Fisher_Stephen}, who determined the fall-off of $ S_2(x_1,x_2) $ 
on the square lattice $ \mathbb{Z}^2 $ for $ K \equiv 1 $ and two monomers in the bulk to behave asymptotically as $ |x_1-x_2|^{-1/2} $ for large separation (making the similarity to the Ising model even more apparent~\cite{PY80}). The values of other special placements of momomer pairs 
on a square lattice are also known (cf.~\cite{Fisher_Stephen,Hartwig,All15}). In case of the infinite half-lattice $ \mathbb{Z} \times \mathbb{Z}_+ $, the monomer boundary correlations in case $ K \equiv 1 $ have been computed not long ago by V.~B. Priezzhev and P.~Ruelle~\cite{PRue08}. They turned out to be Pfaffians with two-point function given by 
\be\label{Ruelle08}
S_2((\xi,0),(\eta,0)) = \begin{cases} - \frac{2}{\pi\, |\xi-\eta|} & \mbox{if $|\xi-\eta| $ is odd} \\ 0 & \mbox{otherwise.}  \end{cases} 
\ee

%\section{Dimer model partition function} 

\section{The double dimer model and its loop gas representation} % and its $n$-point functions} 
\label{Sec:DD}

The removal of a site in a finite graph, or equivalently its cover by a monomer, has a drastic effect on the graph's dimer covers:  if $Z_{\Lambda,K}\neq 0$ then for parity reasons the modified graph has no dimer cover. 
The removal of an even number of sites does not automatically invalidate the existence of a cover.   
Its effect on the distribution of the dimer covers may be localized to a collection of random paths linking pairwise the affected sites.    
A convenient way 
to arrive at such a stochastic geometric picture of correlations 
is to consider the overlay of two sets of dimer covers, one of the original graph and the other of its depleted version resulting in the double dimer model. 
This technique is reminiscent of the duplication which is an effective tool in the study of the Ising model's correlation functions in its random current representation~\cite{GHS,Aiz80}.

\tikzset{
vertexA/.style = {
shape=circle,
draw = black,
fill=white,
very thick
},
vertexB/.style = {
shape=circle,
draw = black,
fill=white,
very thick,
%very thick,
dashed
},
edgeA/.style = {
draw=black,
line width=1.6mm
},
edgeB/.style = {
dashed,
draw=black,
line width=1.6mm
}
}

\begin{figure}[h]
\begin{center}
\begin{tikzpicture}[scale=1.1]
\draw[step=1,gray,very thin] (-2,0) grid (6,3);
\tikzset{every node/.style={vertexA}}
\node (x1) at (0,3) {{\small$x_1$}};
\node (x2) at (0,0) {{\small$x_2$}};
\node (x5) at (4,3) {{\small$x_3$}};
\node (x6) at (4,2) {{\small$x_4$}};

\tikzset{every node/.style={vertexB}}
\node (x3) at (0,1) {{\small$x_5$}};
\node (x4) at (5,1) {{\small$x_6$}};

\tikzset{every edge/.style={edgeA}}
    \draw(1,3) edge (2,3);
    \draw(3,3) edge (3,2);
    \draw(2,2) edge (2,1);
    \draw(1,2) edge (1,1);
    \draw(0,2) edge (x3);
    \draw(1,0) edge (2,0);
    \draw(3,0) edge (4,0);
    \draw(5,0) edge (x4);
    \draw(5,2) edge (6,2);
    \draw(5,3) edge (6,3);
        \draw(-1,0) edge (-1,1);
           \draw(-2,0) edge (-2,1);
    \draw(3,1) edge[bend left] (4,1);
 \draw(-2,3) edge[bend left] (-2,2);
    \draw(-1,3) edge[bend left] (-1,2);
      \draw(6,0) edge[bend left] (6,1);
\tikzset{every edge/.style={edgeB}}
   \draw(x1) edge (1,3);
   \draw(2,3) edge (3,3);
   \draw(3,2) edge (2,2);
   \draw(1,1) edge (2,1);
   \draw(0,2) edge (1,2);

   \draw(x2) edge (1,0);
   \draw(2,0) edge (3,0);
   \draw(4,0) edge (5,0);
   \draw(x5) edge (x6);
   \draw(5,2) edge (5,3);
      \draw(6,2) edge (6,3);
            \draw(-2,0) edge (-1,0);
                 \draw(-2,1) edge (-1,1);

   \draw(3,1) edge[bend right] (4,1);
    \draw(-2,3) edge[bend right] (-2,2);
    \draw(-1,3) edge[bend right] (-1,2);
    \draw(6,0) edge[bend right] (6,1);
\end{tikzpicture}
\caption{A double dimer configuration $\omega^{(2)}= (\omega_1,\omega_2)$  on a finite graph (whose edges are  indicated  in grey) with disjoint sets of monomers covering selected sites in  one or the other copy of the graph.   
%The circles mark monomers which are removed, 
The edges of $ \omega_1$ and $ \omega_2$, are marked in solid and dashed lines, correspondingly.
% mark the edges of  $ \omega_1$, and dashed lines the edges of $ \omega_2 $.
%, and the circles mark monomers which cover the indicated sites on one of the graphs but not the other.   
%
% $ M_1 = \{ x_1,x_2 , x_3, x_4\} $ and $ M_2 = \{ x_5, x_6\} $ respectively. 
 The overlay 
results in a configuration of alternatingly marked paths connecting the monomers and  alternatingly marked loops, as described in Lemma~\ref{lem:loops}. } \label{Fig:Loops}
\end{center}
\end{figure}
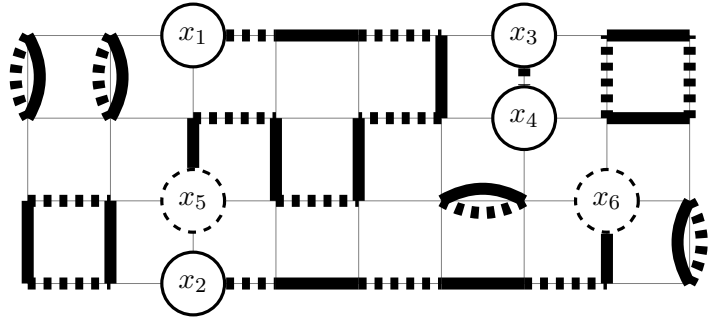

The configurations of doubled dimer covers of a  graph 
$ \G = ( \mathcal{V}, \mathcal{E}) $, depleted by corresponding sets of monomers $M_1, M_2 \subset \mathcal{V}$ will be denoted here as: 
\be 
\omega^{(2)} = (\omega_1, \omega_2) \, \in \,  
\Omega_\G ( M_1) \times \Omega_\G (M_2) \ =:\ \Omega_\G^{(2)} ( M_1, M_2)  \, .  
\ee   
Clearly, each such configuration $\omega^{(2)}$ is in one-to-one correspondence 
with a 2-multigraph with vertex set $\mathcal{V}$ and the collection of edges in $\omega^{(2)}$.  
The following deterministic statement concerning such 
pairs of matchings  relates the duplicated dimer cover model with 
a system 
of loops and paths with prescribed boundaries given by the monomers, cf.~Figure~\ref{Fig:Loops}.

\begin{lemma}[Double matching  as a loop / path system]  \label{lem:loops} 
For any finite graph $\G = ( \mathcal{V}, \mathcal{E})$, let  $\omega^{(2)}  \in \, \Omega_\G^{(2)} ( M_1, M_2)$  be  a  pair of dimer covers of  $\G$ depleted by a disjoint pair of monomers,  $M_1, M_2 \subset \G$.  Then the multiplicity with which the edges are covered by $\omega^{(2)}$ coincides with that of a collection $\Gamma  =\Gamma(\omega^{(2)})$ of edge-disjoint loops and paths 
where each $\gamma\in \Gamma$ is either
\begin{enumerate}[i.]
\item  a double loop covering  a single edge,
\item   a simple loop  of an even number of non-repeated edges,
\item   a simple  path with  boundary set $\partial \gamma \subset  M_1 \sqcup  M_2$.  
\end{enumerate}  
In case~iii., the numbers  of edges of $\gamma$  is odd if and only if its two boundary sites  are in the same monomer set (i.e.\ either both in $M_1$ or both in $M_2$). \\

The loop-path characterisation of double covers $ \omega^{(2)} $ in terms of $\Gamma$ partitions  their collection $\Omega_\G^{(2)} ( M_1, M_2) $  into equivalence classes, each of $2^{n_s(\Gamma)}$ elements, where  $n_s(\Gamma)$ is the number of simple loops in $\Gamma$.  
\end{lemma}

\begin{proof}
 In the case of disjoint monomer sets, the degree of each site $x\in \mathcal{V}$  
in  the multigraph formed from the edge set of $\omega^{(2)}$  is either $1$ or $2$, and given by 
    \be 
        \deg_{\omega^{(2)}}(x) = \deg_{\omega_1}(x) + \deg_{\omega_2}(x) =
        \begin{cases}
            2, &\text{ if }x \in  \mathcal{V} \backslash  [ M_1 \sqcup M_2] \\
            1, &\text{ if }x \in M_1 \sqcup M_2
        \end{cases}.
    \ee
It follows that the collection of edges with multiplicity $1$ is the disjoint union of loops (of no boundary) and paths with end points in $M_1 \sqcup  M_2$, each made of simple edges  in $\omega_j$, at alternating values of $j=1,2$.   The stated constraints on the parity of the number of edges in the loops and paths readily follow from the constraint that the path's edges alternate between the two dimer covers.  In case of the open paths, the identity of the cover to which an edge of $\gamma$ belongs  can be determined successively starting from the end points.   There is no such  constraint for the $n_s(\Gamma)$ simple closed loops, and hence for each of these there are exactly two choices (independent among the loops) for the alternating values of $j\in \{1,2\}$.  
\end{proof}

The above representation of $\Omega_\G^{(2)} ( M_1, M_2) $ in terms of loops and paths may be extended by allowing the two sets of monomers to  overlap, or coincide.   The corresponding pure loop gas  was recently studied in~\cite{Ay15}. \\  
  
The loop gas picture of the double-dimer partition functions  
\be\label{eq:DDPF}
 Z_{\G,K}^{(2)}\left(M_1,M_2\right) \ := \ Z_{\G,K}(M_1)  \, \,  Z_{\G,K}(M_2)  \ = \sum_{\omega^{(2)} \in \Omega^{(2)}(M_1, M_2)} \chi_K(\omega_1) \, \chi_K(\omega_2)  \, , 
\ee
  is particularly convenient in revealing switching symmetries of  the double dimer model's connection amplitudes.  Similar symmetries have been noted for the correlation functions of the Ising model, revealed  there through its random current representation.  

The connection amplitudes are defined as restricted sums such as
\begin{multline}\label{eq:defP}
Z_{\G,K}^{(2)}\left(M_1,M_2;  \connect{x_j}{y_j} \; \mbox{for $ j =1, \dots , N$} \right) : = \\
 \sum_{\omega^{(2)} \in \Omega^{(2)}(M_1, M_2)} \chi_K(\omega_1) \, \chi_K(\omega_2)  \, \prod_{j =1}^N \1\left[  \Gconnect{x_j}{y_j} \right] 
 \, .
\end{multline}
where   $ \{x_j ,  y_j\}_{j=1,..., N} $ are pairs of sites in $M_1\sqcup M_2 $, and $\1\big[  \Gconnect{x_j}{y_j} \big] $ is an indicator function corresponding to the condition that 
 the monomers $ x_j , y_j $ are connected by a path $\gamma \in \Gamma(\omega^{(2)})$. 
\begin{lemma}[Switching principle I]\label{lem:switch}
For any finite graph $\G= (\mathcal{V} , \mathcal{E}) $, pair of disjoint monomer sets  $M_1, M_2$ and  
$ \{x,y\} \subset \mathcal{V}\backslash (M_1 \sqcup M_2)$:
\begin{align}
& Z_{\G,K}^{(2)}\left(M_1 \sqcup \{x,y\}  ,M_2;  \connect{x}{y} , C \right)  = Z_{\G,K}^{(2)}\left(M_1, M_2  \sqcup \{x,y\} ;  \connect{x}{y}  , C\right) \, ,   \\
& Z_{\G,K}^{(2)}\left(M_1 \sqcup \{x\}  ,M_2 \sqcup \{y\} ;  \connect{x}{y} , C \right)  =Z_{\G,K}^{(2)}\left(M_1 \sqcup \{y\}  ,M_2 \sqcup \{x\} ;  \connect{x}{y}  ,C\right)  
\end{align}
where $ C $ stands for any collection of other connection conditions among monomers in $ M_1\sqcup M_2 $. 
\end{lemma}
\begin{proof}  Considering first the case $C = \emptyset$ (i.e. no other conditions),   
let $ \Omega^{(2)}(M_1 \sqcup \{x,y\}  , M_2;\connect{x}{y}) $ be the set of double dimer covers for which there is a path $ \gamma^{(x,y)} \in \Gamma $ with $\partial \gamma^{(x,y)} = \{ x, y\} $. 
The first assertion is based on the  bijection
\be 
 \Omega^{(2)}(M_1 \sqcup \{x,y\}  , M_2;\connect{x}{y}) \   \to \  \Omega^{(2)}(M_1  , M_2 \sqcup \{x,y\} ;\connect{x}{y}) \, \notag 
 \ee 
implemented by the symmetric difference $ \triangle $ of sets: 
\be 
\left(\omega_1,\omega_2\right) \  \mapsto \ \left( \omega_1 \triangle \gamma^{(x,y)}, \omega_2 \triangle \gamma^{(x,y)} \right)  \, .
\ee  
This map reverses the ``edge coloring'' along the  path $  \gamma^{(x,y)}  $ connecting $ x $ and $ y $ with the color indicating to which of the two dimer covers the edge belongs. 
The first identity thus follows immediately from the fact that the path weights are unchanged under a color-flip operation.

The same switching argument implies also the second identity, and the generalization to more general condition $C$. 
\end{proof}

The loop gas formulation of the double dimer model   casts its correlation functions in terms of (discrete) path integrals, thereby bringing  it closer to a broad range of physics models.  
A more explicit version of this representation, which could  be used for an alternative presentation of the analysis which follows,  is stated in Appendix~\ref{app:A}.

\section{Pfaffian structure of boundary monomer correlation functions}

The switching principle allows a simple proof of the fact that 
boundary monomer correlation functions have a Pfaffian nature on all planar graphs. 
The corresponding result for Ising model's boundary spin-spin correlation functions goes back to~\cite{GBK78}.  Our proof parallels the more   recent rederivation of that relation in~\cite{ADTW}.  

For the dimer model the  following statement was derived in~\cite{PRue08} in case of the infinite planar half-lattice for which the two-point function is given by~\eqref{Ruelle08}. 
For other planar graphs, the theorem was recently established by different means in~\cite{GJL16}.

\begin{theorem}[Pfaffian boundary correlations]  \label{thm:boundary}
 For any finite planar  graph $\G= (\mathcal{V} , \mathcal{E}) $ the boundary values of the 
monomer correlation functions satisfy
\begin{align}  
S_{2n}(x_1, ..., x_{2n}) & =  
\sum_{\pi \in \Pi_{2n} } {\rm sgn}(\pi)  \, \prod_{j=1}^{n} 
S_{2}(x_{\pi(2j-1)}, x_{\pi(2j)})  \; \equiv  \Pf_n\left( S_{2}(x_i,x_j) \right)
\end{align}  
where   $M:=\{x_1, ..., x_{2n} \}$   ranges over 
sequences of disjoint vertices  
positioned in a cyclic order along any boundary of $\G$. Moreover, $\Pi_{2n}$ is the collection of pairings of $\{1,...,2n\}$,  and ${\rm sgn}(\pi) $ is the pairing's parity.
\end{theorem} 

\begin{proof}    Through a  known characterization of Pfaffians (provable by an induction argument) it suffices to show that for each $n >1$ and any cyclicly ordered sequence of boundary sites 
\be \label{SQ}
S_{2n}(x_1, ..., x_{2n}) \ = \ Q_{2n}(x_1, ..., x_{2n}) \,  
\ee 
with $Q_{2n}$ defined as: 
\be\label{def:Q}
Q_{2n}(x_1, ..., x_{2n}) \ := \ \sum_{k=2}^{2n} (-1)^k \, S_{2}(x_1,x_k)\,  S_{2(n-1)} \, .  (\cancel{x_1},x_2, ..., \cancel{x_k}, ..., x_{2n})
\ee 
At fixed  $k$ the term $S_{2}(x_1,x_k)\,  S_{2(n-1)} \,  (\cancel{x_1},x_2, ..., \cancel{x_k}, ..., x_{2n}) $ is a   sum of over configurations of the duplicated system,  $\omega^{(2)} \in \Omega^{(2)}(\{x_1,x_k\}, \{\cancel{x_1},x_2, ..., \cancel{x_k}, ..., x_{2n} \})$, which may be grouped 
according to the paths of $\Gamma(\omega^{(2)})$ which connect to $x_1$ and $x_k$.   These fall into two classes: the  monomers $x_1$ and $x_k$ may be connected to each other by some $\gamma \in \Gamma$, or else each is connected to another  monomer:
\begin{align}  \label{eq:Qexpand}
& Q_{2n}(x_1, ..., x_{2n}) \left( Z_{\G,K}\right)^2   \ =  \\
&   \ \sum_{k=2}^{2n} (-1)^k  \, Z_{\G,K}^{(2)}(\{ x_1,x_k\}  , \{ \cancel{x_1},x_2, ..., \cancel{x_k}, ..., x_{2n} \}  ;  \connect{x_1}{x_k} ) \notag \\
& \quad  +   \sum_{k=2}^{2n}  (-1)^k  \, \sum_{\substack{l,m=2 \\ k\neq l \neq m \neq k} }^{2n}   \, Z_{\G,K}^{(2)}\left(\{ x_1,x_k\}  , \{ \cancel{x_1},x_2, ..., \cancel{x_k}, ..., x_{2n} \}  ; { \connect{\ x_1}{x_m}\atop\connect{x_k}{x_l} } \right) \, . \notag 
\end{align}
 
Being based on  combinatorial arguments, the above relation holds for arbitrary graphs.   It will now be combined with the following  topological implication of planarity.   For any planar graph,  a pair of  monomers $ \{x_i,x_j\} $ located along  the boundary can be linked by one of the non-intersecting simple paths of $\Gamma(\omega^{(2)})$ only if the two are either consecutively placed along the boundary or separated by an even number of other monomers.  In other words, in the cases considered here: 
\be 
 \connect{x_i}{x_j}  \ \Longrightarrow \  (-1)^{i-j} = -1 \,. 
\ee

For the pair of sums on the  right side of~\eqref{eq:Qexpand} this implies:  
\begin{enumerate} [i.] 
\item  In the first sum $ (-1)^{k-1} = -1 $, and hence
 \begin{align}
 & \sum_{k=2}^{2n} (-1)^k  \, Z_{\G,K}^{(2)}(\{ x_1,x_k\}  , \{ \cancel{x_1},x_2, ..., \cancel{x_k}, ..., x_{2n} \}  ;  \connect{x_1}{x_k} ) = \notag \\
& \sum_{k=2}^{2n} (-1)^k \, Z_{\G,K}^{(2)}(\emptyset , M ;  \connect{x_1}{x_k} )  = \sum_{k=2}^{2n} Z_{\G,K}^{(2)}(\emptyset , M ;  \connect{x_1}{x_k} )  \notag \\  
&\qquad  = 
Z_{\G,K}(M) \, Z_{\G,K} = S_{2n}(x_1,\dots, x_{2n})\,  \left( Z_{\G,K}\right)^2   \, . 
\end{align}
Here the first step is a consequence of the switching principle of Lemma~\ref{lem:switch}. 

\item In the second sum $ (-1)^{k-l} = -1$, and thus 
\be 
\sum_{\substack{k,l=2 \\m\neq k \neq l \neq m} }^{2n}  (-1)^k  \, Z_{\G,K}^{(2)}\left(\{ x_1,x_k\}  , \{ \cancel{x_1},x_2, ..., \cancel{x_k}, ..., x_{2n} \}  ; { \connect{\ x_1}{x_m}\atop\connect{x_k}{x_l} } \right) \\
=  \ 0
\ee 
due to  the antisymmetry of the summands under  the exchange of $k$ with $l$ as is apparent from the switching principle of Lemma~\ref{lem:switch}.
\end{enumerate}  

Upon insertion in \eqref{eq:Qexpand} these relations  prove \eqref{SQ}, and through it the claimed Pfaffian structure. 
\end{proof} 

\section{Disorder operators for the dimer model} 

%, which may be amorphous

 In the context of  planar Ising spin systems order-disorder correlation functions  have a Pfaffian structure throughout the bulk  and reduce to  simple correlations functions in case of sites along the boundary.
They have been recently discussed, from a pair of somewhat different perspectives, in \cite{CCK16} and \cite{ADTW}. 
To present a related  concept for the dimer model's correlation functions we turn now to  the dimer analog of disorder operators.  

The definition of the disorder operators may be placed in  the broader context of gauge symmetries.  For that let us first recall  Kasteleyn's observation~\cite{Kas63} that the dimer model has the following  $Z_2$ gauge symmetry  in the dependence of the partition function $Z_{\G,K}$ on the kernel $K$. \\ 

For  subsets $B\subset \mathcal{V} $ let us denote  
\be 
\partial B := \{ [x,y] \in \mathcal{E} \, | \,  \mbox{if exactly one of the two points is in $B$} \} \, 
\ee  
which forms the   {\it edge boundary} of $B$. \\

Next, for any edge set $E \subset \mathcal{E} $  let  
$T_{E }: \mathbb{C}^{\mathcal{E} }\to  \mathbb{C}^{\mathcal{E} } $ be the transformation of $K$  which flips its  signs  over the edges in $E$, 
\be
(T_{E}K)_b \ = \  \begin{cases}  -  K_b  & \mbox{if} \; b  \in E \\ 
  ~~K_b  & \mbox{otherwise.}  
  \end{cases} 
\ee

The key observation now is that if $E = \partial B$ for a set $B\subset \mathcal{V} $  then 
\be \label{eq:Z2gauge}
Z_{\Lambda,T_{\partial B} K}  \ = \ (-1)^{|B|}\,  Z_{\Lambda,K}  \, ,
\ee 
where   $|B | $ is the number of sites in $B$. 
For  $B$ containing a single site  the relation \eqref{eq:Z2gauge} holds since in each dimer cover exactly one dimer is affected by the sign flip $T_{\partial B} $.  
The general case follows by noting the commutative product relation 
\be T_{\partial B} =  \prod_{x\in B} T_{\partial \{x\}} \   
\ee
and taking the corresponding product of  the single site case of \eqref{eq:Z2gauge}.\\ 

In view of the simplicity of the effect of $T_{\partial B}$ on the partition function (and also on the expectations defined below), such mappings may be regarded as the model's gauge transformations. \\   

The disorder operators  which are  defined next may be viewed as partial gauge  transformations,  given by $T_E$ where $E$ is the collection of edges which are traversed by a line $\ell $ which has only transversal intersections with the edges of $\mathcal {E}$ and in the general case has a non-empty boundary set 
$\partial \ell$.  The end-points of $\ell$ are associated with sites of the dual graph $ \G^* $, namely the faces of $ \G $ in which the end points of $\ell $ lie.    One may note that away from  $\partial \ell$  the transformation locally acts as if it could be associated with a gauge transformation -- but it is not (unless $\partial \ell = \emptyset$).

\begin{definition} \label{def:disorder}
For a planar graph $\G=(\mathcal{V},\mathcal{E})$ with edge weights $ K \,: \mathcal E \mapsto \C$:
\begin{enumerate}[i.]
\item The {\it disorder operators} $\tau_{\ell}$ 
are associated with site-avoiding, lines $ \ell_1, \dots,  \ell_n $ in the plane in which $\G$ is embedded.  To each  such line  
we associate the transformation 
$ K \mapsto T_{\ell^*} K$  where 
$\ell^*$ is the set of edges in $\mathcal {E}$ which are crossed  by  $\ell$ an odd number of times.

%the {\it disorder operators} $\tau_{\ell_j}$ are associated with open-ended, site-avoiding lines $ \ell_1, \dots , \ell_n $ in the plane in which $\G$ is embedded. These lines 
%give rise to a partial gauge transformations
%$ K \mapsto T_{\ell_j^*} K$  where for each line $\ell$, as above, 
%$\ell^*$ is the set of edges in $\mathcal {E}$ which are crossed  by  $\ell$ an odd number of times.  
\item The  {\it expectation values} of products of such disorder operators is defined as: 
\be \label{ord_disord_corr} 
\langle \prod_{j=1}^n \tau_{\ell_j} \rangle_{\G,K}  \ := \  \frac{Z_{\G, T_{\ell_1^*} \circ \dots \circ T_{\ell_n^*} K} }{Z_{\G, K} }\,.  
\ee 
\end{enumerate}
\end{definition} 

As an expression of the above mentioned gauge symmetry, the expectation value $ \langle \prod_{j=1}^N \tau_{\ell_j} \rangle_K $ is a homotopy invariant  under deformations of  any $\ell_j$ in the plane which preserve the line's  endpoints.  More precisely, as a simple consequence of~\eqref{eq:Z2gauge} we have: 

\begin{proposition} [Homotopy invariance]
For any finite planar graph $\G=(\mathcal{V},\mathcal{E})$, edge weights $ K \,: \mathcal E \mapsto \C$ and lines  $ \ell_j $, $ j \in \{1,\dots, n\} $, as in Definition~\ref{def:disorder}, under deformations of  each $\ell_j$ in the plane which preserve the line's  endpoints the expectation value functional $ (\ell_1,\dots ,\ell_N) \mapsto  \langle \prod_{j=1}^n \tau_{\ell_j} \rangle_{\G,K} $  is multiplied by $(-1)$ each time one deformed line is moved over a site of the planar graph. 
\end{proposition} 

%$(x_j, x^*_j)$
%
%$\ell_j$ 
%
%$(1)

 \begin{figure}[h]
\begin{center}
\includegraphics [width = 1.1\textwidth]{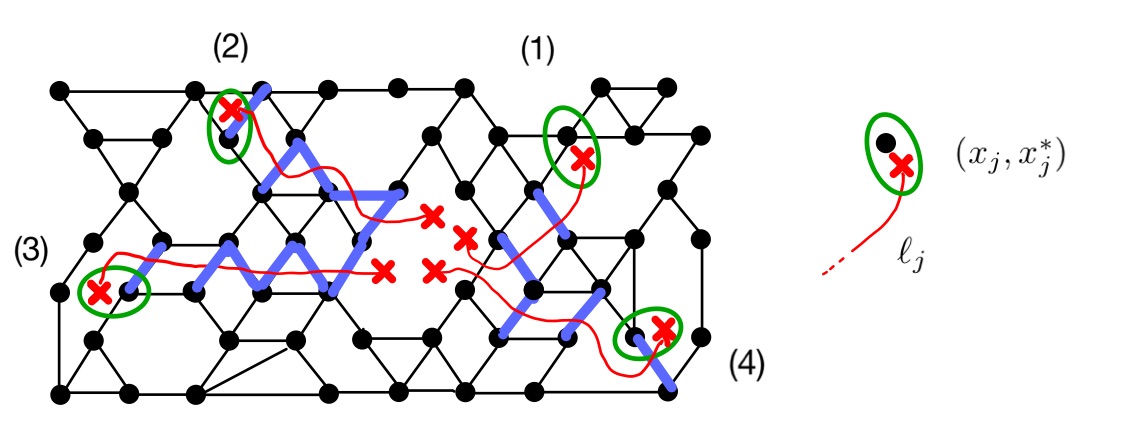} %{tau_variables} 
\caption{Order-disorder variables for a planar graph.  
Each of the ovals in the figure encircles a pair consisting of a site  $x_j\in \G$ and a point, marked $\times$,  within an adjacent cell of the dual graph  $x_j^*\in \G^*$.   
The disorder variables $\tau_{\ell_j}$ are associated with  lines $\ell_j $, each linking the corresponding $\times$ marked sites with a point in the {\it grand central} cell  $x_0^* $.   
%Each of the  monomer sites $x_j$ is adjacent to the dual site $x_j^* $ of the corresponding disorder variable.
%
%which is a neighbor of $x_j$ in $\G \times \G^*$  with a common dual site $x_0^*\in \G^*$, called grand central. 
The disorder lines $ \ell_1, \ell_2, \dots $ are enumerated cyclicly in the order of the lines' emergence from the grand central $x_0^* $.    The correlation function associated with such an array is defined in \eqref{ord_disord_corr}\label{fig:ord_disord}   
}
\end{center}
\end{figure}

The above construction parallels the definition of disorder operators for the Ising model~\cite{Kad_Ceva}.   Disorder lines for the dimer-monomer model appear also in the recent discussion of the dimer model's partition function in terms of Grassmann integrals~\cite{AllF14}.  

 \section{Pfaffian structure of the correlation functions of the order-disorder operators}

 Our main concern in this paper will be canonical pairs of order-disorder variables, cf.~Figure~\ref{fig:ord_disord}. 

 \begin{definition}
For a planar graph $\G=(\mathcal{V},\mathcal{E})$ with a set of edge weights $ K \,: \mathcal E \mapsto \C$,  open-ended, site-avoiding, non-intersecting lines $ \ell_1, \dots , \ell_{2n} $  in the plane in which $\G$ is embedded, together with disjoint sites $ x_1, \dots , x_{2n} \subset \mathcal{V} $   are called a collection of \emph{canonical pairs of order-disorder variables} in case:
\begin{enumerate}[i.]
\item all lines have a common end-point $ x_0^* \in \G^* $, called 
the \emph{grand central}, and 
\item the other endpoint of $ \ell_j $ is a face $ x_j^* \in \G^* $ adjacent to $ x_j $ for all $ j \in \{ 1, \dots , 2n \} $. 
\end{enumerate}
We call the canonical pairs of order-disorder variables \emph{cyclicly ordered} if they are labeled relative to their intersections with  the edge boundary of $ x_0^*$. 

The  {\it expectation values} of products of order-disorder variables operators $ \mu_j := \eta_{x_j} \tau_{\ell_j} $ are defined as 
\be 
%T_{2n}( p_1,\dots, p_{2n}) := 
\langle \prod_{j=1}^{2n} \mu_j \rangle_{\G,K}  \ := \  \frac{Z_{\G, T_{\ell_1^*} \circ \dots \circ T_{\ell_{2n}^*} K}\left(\{ x_1,\dots,x_{2n}\}\right) }{Z_{\G, K} }\, .
\ee 
%where $ p_j := (x_j , \ell_j) $ stands for an order-disorder variable. 
 \end{definition}

 Our main new result is:

\begin{theorem}[Pfaffian correlations] \label{thm:Pf_OD}  
For a finite planar graph $\G=(\mathcal{V},\mathcal{E})$ with edge weights $ K \,: \mathcal E \mapsto \C$, for any collection of canonical pairs of  
order-disorder variables $p_j= (x_j,\ell_j) $,  $ j \in \{ 1, \dots , 2n\} $, ordered cyclicly relative to the grand central 
\be  \label{eq:Pf_gen} 
 \langle  \prod_{j=1}^{2n} \mu_j  \rangle_{\G,K} 
\ = \  \sum_{ \pi \in \Pi_n} \sgn(\pi) \, \prod_{j=1}^n  
\langle \mu_{ \pi(2j-1)} \,  \mu_{ \pi(2j)}  \rangle_{\G,K}   \ \equiv \ \Pf_n\left( \langle \mu_j \mu_k   \rangle_{\G,K}  \right)  \, .   
\ee
\end{theorem} 

This result includes Theorem~\ref{thm:boundary}  as a special case.   To see that, let us first note that 
for sites $x_j$ which lie along the boundary of the grand-central $x_0^*$, the corresponding disorder sites may be chosen as $x_j^* = x_0^*$.   When the lines $\ell_j$  do not cross any edge, as in this case,  the operators $\tau_{\ell_j} $ act as identity and may be omitted.    Theorem~\ref{thm:boundary}  then emerges through the inverted picture of the plane in which the complement of the finite graph is viewed as a single cell (of potentially large boundary).\\

In case the monomers  $ \{ x_{2j-1},x_{2j} \}$ are pairwise adjacent, 
the disorder lines may be chosen so that their actions are pairwise equivalent, and thus cancel each other.  In that case the pairwise product of two order-disorder variables reduces to a  an ordinary product of monomers, i.e., a dimer $ 
\mu_{2j-1} \mu_{2j}  =  \eta_{x_{2j-1}} \eta_{x_{2j}}$, so that  
\be
 \big\langle  \prod_{j=1}^{2n} \tau_j  \big\rangle_{\G, K}  =  \big\langle  \prod_{j=1}^{n}\eta_{x_{2j-1}} \eta_{x_{2j}} \big\rangle_{\G, K} \, .
\ee
 \\

The proof of Theorem~\ref{thm:Pf_OD} is organized along the lines used to establish the boundary case, Theorem~\ref{thm:boundary}.   However, the relevant topological considerations are considerably more intricate.   Defining, in analogy with $ Q_{2n} $ of~\eqref{def:Q},   
\be 
R_{2n}(p_1,...p_{2n}) \ := \ \sum_{k=2}^{2n} (-1)^k\,  \langle  \mu_1 \mu_k  \rangle_{\G,K} \  \ 
 \langle  \prod_{j\in \{\cancel 1, 2,..,\cancel k,..,2 n\}} \mu_j  \rangle_{\G,K}  \, ,
\ee
(with $ p_j := (x_j , \ell_j) $ standing for an order-disorder variables)
% (which is analogous  to  $ Q_{2n} $ of~\eqref{def:Q}  but with $ S $ replaced by $ T $)
 the Pfaffian structure will be shown by proving  that for each $n$ and choice of order-disorder pairs:
\be 
R_{2n}(p_1,...p_{2n})  \ = \ \langle  \prod_{j\in \{ 1,..,  2 n\}} \mu_j  \rangle_{\G,K}  \, . 
\ee 
At specified $k$ the product of the order-disorder correlators is given by: 
\begin{align} \label{R1} 
&\langle  \mu_1 \mu_k  \rangle_{\G,K} \  \ 
 \langle  \prod_{j\in \{\cancel 1, 2,..,\cancel k,..,2 n\}} \mu_j  \rangle_{\G,K}  \ 
 \times  Z_{\G,K}^2  \ = \  \\ 
& \quad \sum_{\omega^{(2)} \in \Omega^{(2)}(\{x_1,x_k\}, (\{\cancel {x_1}, x_2,..,\cancel {x_k},.., x_{2 n}\}}  \chi_K(\omega_1) \, \chi_K(\omega_2) \, (-1) ^{(\omega_1 |\ell_{1,k})}  (-1) ^{(\omega_2 | \mathcal L \backslash  \ell_{1,k})} \notag
\end{align} 
where $(\omega_1 |\ell_{1,k}) $ denotes the number of intersections of the edges of $ \omega_1 $ with two disorder lines $ \ell_{1,k} := \{ \ell_1, \ell_k \} $ and likewise  $(\omega_2 | \mathcal L) $ denotes the number of intersections of the edges of $ \omega_2 $ with the collection of all disorder lines $ \mathcal L := \{ \ell_1, \dots , \ell_{2n} \} $.  

The terms in the above sum can be split into two classes, according to whether the loop / path configuration $\Gamma(\omega^{(2)})$ includes a path with $\partial \gamma^{(1,k)} = \{x_1,x_k\}$, or not.   The corresponding partial sums will be studied through the following quantities:
\begin{multline}  \label{def:W}
W^{(2)}_{\G, K}(\{M_1,\mathcal{L}_1\} ,\{ M_2,\mathcal{L}_2\}  ; C ) :=  \\
 \sum_{{\omega^{(2)} \in \Omega^{(2)}(M_1, M_2)}} \1\left[ \mbox{$ \omega^{(2)} $ satisfies $ C $}\right]  \, \Chi_{K} (\omega_1) \, (-1)^{(\omega_1 \, | \, \mathcal{L}_1)}\,  \Chi_{K}(\omega_2) \, \, (-1)^{(\omega_2 \, | \, \mathcal{L}_2)}
 \end{multline}
in which we specify a set of connections $ C $ of the involved monomer sets $M_1, M_2$. \\

A key result here is the corresponding version of the switching lemma: 
\begin{lemma}[Switching principle II]\label{pr:switch2}
For planar graphs, and the setup of Theorem~\ref{thm:Pf_OD},  we have for any $ m \neq k \neq l \neq m $:
\begin{align}
 &  W^{(2)}_{\G, K}(\{p_1, p_k \} ,\{ \cancel{p_1}, p_2 , \dots , \cancel{p_k} , \dots ,  p_{2n}  \}; \connect{x_1}{x_k} )    \notag  \\
 & \quad = (-1)^k
  \, W^{(2)}_{\G, K}(\emptyset ,\{ p_1,\dots , p_{2n}  \}; \connect{x_1}{x_k} )   \label{R1}
 \\[2ex] 
   &  W^{(2)}_{\G, K}\left( \{p_1, p_k \} ,\{ \cancel{p_1}, p_2 , \dots , \cancel{p_k} , \dots ,  p_{2n}  \}; { \connect{\ x_1}{x_m}\atop\connect{x_k}{x_l} } \right) \notag  \\
 & \quad =  (-1)^{k-l-1}  \,    W^{(2)}_{\G, K}\left( \{p_1, p_l \} ,\{  \cancel{p_1}, p_2,... , \dots , \cancel{p_l} , 
 \dots ,  p_{2n}  \}; { \connect{\ x_1}{x_m}\atop\connect{x_k}{x_l} } \right) 
 \label{R2}
\end{align}
\end{lemma}

\begin{proof}  
 The relation \eqref{R1}, which involves terms for which $\connect{x_1}{x_k}$, 
will be established through the switching transformation:  
\be
 (\omega_1,\omega_2) \ \mapsto \  (\omega_1\Delta \gamma^{(1,k)} ,\omega_2\Delta \gamma^{(1,k)} ) 
 \ee 
Expanding  the  quantities $W^{(2)}$ (defined in \eqref{def:W}), which appear in \eqref{R1}, into sums over $\omega^{(2)} $, 
 the ratio of the corresponding terms  is
\begin{multline} \label{Z2}
\frac{\chi_K(\omega_1\Delta \gamma^{(1,k)} ) \ \  \chi_K(\omega_2\Delta \gamma^{(1,k)} ) }{\chi_K(\omega_1) \quad  \chi_K(\omega_2) } \, 
\frac{  (-1) ^{(\omega_2\Delta \gamma^{(1,k)}   | \mathcal L )} 
}
{ (-1) ^{(\omega_1 |\ell_{1,k})}  (-1) ^{(\omega_2 | \mathcal L \backslash  \ell_{1,k})} 
 }   \\
 \ = \  (-1)^{(\gamma^{(1,k)} |\mathcal L)} \ \ (-1)^{(\omega^{(2)} | \ell_{1,k})}\, , 
\end{multline} 
where the last step is by an elementary calculation in $Z_2$.   The  relation \eqref{R1} then follows from the special case $ l = 1 $ through the  lemma which is stated next. (This is  where the model's planarity plays a role.) \\ 
  
 The relation  \eqref{R2} concerns terms $\omega^{(2)}$ for which   $\connect{x_k}{x_l}$ for some $l\neq 1$.  For that  we employ the switching transformation
\be
 (\omega_1,\omega_2) \ \mapsto \  (\omega_1  \Delta \gamma^{(k,l)} , \omega_2 \Delta \gamma^{(k,l)}) \,. 
 \ee 
By a calculation similar to \eqref{Z2}, the ratio of the corresponding contributions to the  sums which yield the two quantities $W^{(2)}$  in \eqref{R2} is: 
\begin{align}
\frac{(-1)^{(\omega_1 \Delta \gamma^{(k,l)} | \ell_{1,l} )} \, (-1)^{(\omega_2  \Delta \gamma^{(k,l)} |  \mathcal{L} \backslash \ell_{1,l} )}  }{(-1)^{(\omega_1 | \ell_{1,k} )} \, (-1)^{(\omega_2  |  \mathcal{L} \backslash \ell_{1,k} )} } \ & = \ \frac{ (-1)^{(\gamma^{(k,l)} |\mathcal L)} \, (-1)^{(\omega^{(2)} | \ell_{1,l} )}}{(-1)^{(\omega^{(2)} | \ell_{1,k} )}} \notag \\
& = \ (-1)^{(\gamma^{(k,l)} |\mathcal L)} \, (-1)^{(\omega^{(2)} | \ell_{k,l} )} \, . 
\end{align}
The  relation \eqref{R2} then again follows from the next lemma. 
\end{proof} 

The topological statement which was quoted within the above proof is: 

\begin{lemma}[Intersection parities] In  the planar graph setup of  Proposition~\ref{pr:switch2}, for any
$\omega^{(2)}$ such that $\connect{x_k}{x_l}$ with respect to the corresponding loop / path configuration  $\Gamma(\omega^{(2)})$:
\be  \label{RS1}
  (-1)^{(\gamma^{(k,l)} |\mathcal L)} \ \ (-1)^{(\omega^{(2)} | \ell_{k,l})} \ = \  (-1)^{k-l-1} \, . 
\ee
\end{lemma}

\begin{proof} 
To establish this relation it is useful to join the open ended paths $\gamma$  of $ \Gamma(\omega^{(2)} ) $  with the  disorder lines corresponding to the paths' edges into loops with only transversal crossing.   For this purpose, we employ the following construction. 
\begin{enumerate}[1.]
\item Join directly each $ x_j  $ with the endpoint~$ x_j^*$ of the corresponding disorder line~$ \ell_j $. 
\item Connect pairwise the other endpoints of the disorder lines within the grand central $ x_0^* $, so that $ \ell_k $ is connected to $ \ell_l $ and the remaining lines are paired consecutively with respect to the cyclic ordering. 
\end{enumerate}

Let $ \sigma^{(k,l)} $  be the loop which includes $ \gamma^{(k,l)} $ concatenated with $ \ell_k $ and $ \ell_l $ in the above construction, and let $ \Sigma^{(k,l)} $ stand for the collection of  the other loops which the construction yields.   Any two planar loops, simple or not,  with transversal crossings can  intersect only even number of times (as can be deduced  from the Jordan curve theorem).   Thus $ \sigma^{(k,l)} $ has an even intersection with $ \Sigma^{(k,l)} $.    The  intersections within the grand central cell contribute to this the factor $(-1)^{k-l-1}$, and the rest is the parity of  the intersections of 
$ \gamma^{(k,l)} $  and $\ell_{k,l}$ with the rest.  Hence: 
\begin{eqnarray} 
1  & = &    (-1)^{k-l-1} \, (-1)^{\left[(\gamma^{(k,l)} | \mathcal L) - (\gamma^{(k,l)} |  \ell_{k,l} ) \right]}  \, 
(-1) ^{\left[ (  \omega^{(2)}| \ell_{k,l} ) - (  \gamma^{(k,l)} | \ell_{k,l} ) \right]}  \notag  \\[1ex]  
& =&   (-1)^{k-l-1} \,   (-1)^{(\gamma^{(k,l)} |\mathcal L)} \ \ (-1)^{(\omega^{(2)} | \ell_{k,l})}\, ,  
\end{eqnarray}
as claimed in \eqref{RS1}. 
\end{proof}

We are now ready to complete the proof of Theorem~\ref{thm:Pf_OD}. 

\begin{proof}[Proof of Theorem~\ref{thm:Pf_OD}]
Similarly as in the Proof of Theorem~\ref{thm:boundary}, it remains to show that 
\be
%T_{2n}(p_1,\dots,p_{2n}) 
\langle \prod_{j=1}^{2n} \mu_j \rangle_{\G,K}  \ =\  R_{2n}(p_1,...p_{2n})  \, . 
\ee
The right side times $ (Z_{\G,K})^2 $ may be rewritten as
\begin{align}
&  \sum_{k=2}^{2n} (-1)^k  \, W^{(2)}_{\G, K}(\{p_1, p_k \} ,\{ \cancel{p_1}, p_2 , \dots , \cancel{p_k} , \dots ,  p_{2n}  \}; \connect{x_1}{x_k} )  \notag \\
&\quad  +   \sum_{k=2}^{2n}  (-1)^k  \, \sum_{\substack{l,m=2 \\ k\neq l \neq m \neq k} }^{2n}   \,  W^{(2)}_{\G, K}\left( \{p_1, p_k \} ,\{ \cancel{p_1}, p_2 , \dots , \cancel{p_k} , \dots ,  p_{2n}  \}; { \connect{\ x_1}{x_m}\atop\connect{x_k}{x_l} } \right) \notag \\
& = \  \sum_{k=2}^{2n}   W^{(2)}_{\G, K}(\emptyset ,\{ p_1,\dots , p_{2n}  \}; \connect{x_1}{x_k} ) \, .  \label{eq:proofpf}
\end{align}
Here the last line results from the switching Lemma~\ref{pr:switch2}. More precisely, 
the second sum on the left vanishes thanks to  the antisymmetry in the $ k\neq l $ summation as is apparent from \eqref{R2}. Applying \eqref{R1} to the first sum on the left yields the sum on the right, which coincides with $ T_{2n}(p_1,\dots,p_{2n})  $. 
\end{proof}

\appendix 

\section{A path integral representation}\label{app:A}

The loop gas formulation of the double dimer model, which is presented in Section~\ref{Sec:DD},  is of help in relating it to a broad range of physics models, for which related techniques are of relevance.     To  highlight this picture, let us just state here the resulting  path integral representation (in a discrete sense)   of  the model's correlation function.

Lemma~\ref{lem:loops} allows  to classify the double-dimer cover configurations in terms of   the loop-gas configuration $\Gamma( \omega^{(2)})$.  Upon partial summation in~\eqref{eq:DDPF}  over the equivalence classes of configurations with common $\Gamma( \omega^{(2)})$ one gets 
\be \label{eq:Gamma1}
 Z_{\G,K}^{(2)}(M_1,M_2) 
\ = \  
 \sum_{{\Gamma \in \Omega_\G^{(L)}(M_1, M_2)}} 
  2^{n_s(\Gamma)}\, \prod_{\gamma \in \Gamma} \chi_K(\gamma) 
\ee 
where $\Omega_\G^{(L)}(M_1, M_2)$ is the collection of loop / path configurations which are consistent with the conditions listed in Lemma~\ref{lem:loops}, and $ \chi_K(\gamma)  = \prod_{b \in \gamma} K_b $ for each $\gamma \in \Gamma$.   Next, summing over the loops of $\Gamma$, while keeping fixed the configuration's  the open-ended paths, one obtains a path representation of the monomer correlation functions.

For  the monomer correlation function, which is defined in~\eqref{M_ratio}, this yields
\begin{eqnarray} \label{eq:two-pointCF}
 S_{2}(x_1,x_{2}) & =&   (Z_{\G,K})^{-2}    \sum_{{\Gamma \in \Omega_\G^{\text{L}}(\{x_1,x_2\}, \emptyset)}} 
  2^{n_s(\Gamma)}\, \prod_{\gamma \in \Gamma} \chi_K(\gamma) \, \notag  \\
 & =&  \sum_{\substack{\gamma  \in \Omega^{\rm A}_1 \\  \partial \gamma= \{x_1,x_2\}} } \chi_K(\gamma) \,  \left(\frac{Z_{\G,K}(\mathcal{V}(\gamma))}{Z_{\G,K}}\right)^2 \1\left[ \gamma \, \mbox{is odd} \right]\, , 
 \end{eqnarray} 
where $ \Omega_1^{\rm A} $ denotes the collection of simple paths on $\G$.

For a more general expression we use  $\Gamma_P$ to refer to  collections of non-intersecting simple paths on the graph $\G$, and denote by  $ \Omega_n^{\rm A} $ the set of such path collections of $n$ elements.   The set of vertices which are covered by paths in $\Gamma_P$ will be denoted by 
 $\mathcal{V}(\Gamma_P)$, 
and the collection of the paths' boundary points by
$\partial \Gamma_P =  \sqcup_{\gamma \in \Gamma_P} \partial \gamma  $.

In these terms,  \eqref{eq:Gamma1} yields the following  path representation.
 \begin{proposition}[Path integral for correlations]
For any finite graph $\G= (\mathcal{V} , \mathcal{E}) $ and disjoint sites $ \{ x_1, \dots , x_{2n} \} \subset \mathcal{V} $ the monomer correlation function admits the  representation
\be  \label{1st_sum}
S_{2n}(x_1, \dots , x_{2n})  \ = \    
\sum_{\substack{\Gamma_P = \{ \gamma_1, ..., \gamma_n\} \subset \Omega_n^{\rm A}  \\ 
\partial \Gamma_P = \{x_1,...,x_{2n}  \}}} 
   w_K(\Gamma_P) \, \prod_{\gamma \in \Gamma_P}  \1\left[ \gamma \,  \mbox{\rm is odd} \right] \, ,
\ee 
with  the weight function
\be
w_K(\Gamma_P) := \  \left(\frac{Z_{\G,K}( \mathcal{V}(\Gamma_P))}{Z_{\G,K}}\right)^2  
\prod_{\gamma \in \Gamma_P}  \chi_K(\gamma) \, . 
\ee
\end{proposition}

\section*{Acknowledgements}  This work was supported in part by the NSF grant PHY-1305472.   We thank Hugo Duminil-Copin and Vincent Tassion for stimulating discussions of related topics, and Princeton University for hosting S.~Warzel as a PU Global Scholar.

 \end{document}